\pdfoutput=1
\documentclass[a4paper,11pt]{article}
\usepackage{amsthm}

\usepackage[
  textheight=236mm,  
  textwidth=156mm,
]{geometry}

\usepackage{amsmath, amssymb, amsfonts}
\usepackage{bbm}
\usepackage{latexsym}
\usepackage{color}
\usepackage{soul}
\usepackage{cite}
\usepackage{xcolor}
\usepackage[colorlinks=false]{hyperref} 
\hypersetup{pdfstartview=FitH,pdftoolbar=false,bookmarks=false} 
\usepackage{microtype}
\usepackage{cleveref}

\newcommand{\R}{\mathbb{R}}
\newcommand{\Q}{\mathbb{Q}}
\newcommand{\Z}{\mathbb{Z}}
\newcommand{\N}{\mathbb{N}}
\newcommand{\allone}{\mathbbm{1}}
\newcommand{\conv}{\mathop\mathrm{conv}\nolimits}

\renewcommand{\subset}{\subseteq}
\renewcommand{\supset}{\supseteq}
\newcommand{\isclaim}{\renewcommand{\qedsymbol}{
  $\lozenge$
}}

\newtheorem{theorem}{Theorem}

\newtheorem{corollary}{Corollary}

\newtheorem{claim}{Claim}
\crefname{theorem}{Theorem}{Theorems}
\crefname{lemma}{Lemma}{Lemmas}
\crefname{corollary}{Corollary}{Corollaries}
\crefname{proposition}{Proposition}{Propositions}
\crefname{claim}{Claim}{Claims}
\crefname{equation}{}{}
\crefrangelabelformat{equation}{\textup{(#3#1#4)--(#5#2#6)}}   %
\crefname{section}{Section}{Sections}
\usepackage{enumitem}
\newlist{myenumerate}{enumerate}{1}
\setlist[myenumerate]{
  label={\itshape(\roman*)},  
  ref={\itshape(\roman*)},    
}
\crefname{myenumeratei}{}{}

\begin{document}

\title{
  On the Complexity of Recognizing Integrality and Total Dual Integrality of the $\{0,1/2\}$-Closure
}

\author{%
  Matthias Brugger%
  \thanks{Operations Research, Department of Mathematics, Technische Universit\"at M\"unchen, Germany. \newline 
  E-mail: \texttt{\{matthias.brugger,andreas.s.schulz\}@tum.de}
  \newline
  Supported by the Alexander von Humboldt Foundation with funds from the German Federal Ministry of Education and Research (BMBF).}
  \and
  Andreas S. Schulz%
  \footnotemark[1]
}

\date{}

\maketitle

\begin{abstract}
The $\{0,\frac{1}{2}\}$-closure of a rational polyhedron $\{ x \colon Ax \le b \}$ is obtained by adding all
Gomory-Chv\'atal cuts that can be derived from the linear system $Ax \le b$ using multipliers in $\{0,\frac{1}{2}\}$.
We show that deciding whether the $\{0,\frac{1}{2}\}$-closure coincides with the integer hull is strongly NP-hard.
A direct consequence of our proof is that, testing whether the linear description of the $\{0,\frac{1}{2}\}$-closure 
derived from $Ax \le b$ is totally dual integral, is strongly NP-hard.
\end{abstract}


\section{Introduction}

Let $P = \{ x \in \R^n \colon Ax \le b \}$ with $A \in \Z^{m \times n}$ and $b \in \Z^m$ be a rational polyhedron.
The integer hull of $P$ is denoted by $P_I = \conv(P \cap \Z^n)$.
Any inequality of the form
$
  u^T Ax \le \left\lfloor u^T b \right\rfloor
$
where $u \in \R^m_{\ge 0}$ and $u^T A \in \Z^n$
is valid for $P_I$. Inequalities of this kind are called \emph{Gomory-Chv\'atal cuts} for $P$ \cite{Gom,Chv}.
The intersection of all halfspaces corresponding to Gomory-Chv\'atal cuts yields the
\emph{Gomory-Chv\'atal closure} $P'$ of $P$. In fact, $[0,1)$-valued multipliers $u$ suffice (see, e.g., \cite{CCZ}), i.e.,
\[
  P' = \{ x \in P \colon u^T Ax \le \left\lfloor u^T b \right\rfloor,\,
  u \in [0,1)^m,\, u^T A \in \Z^n \}.
\]
Caprara and Fischetti \cite{CF} introduced the family of Gomory-Chv\'atal cuts with multipliers $u \in \{0,\frac{1}{2}\}^m$.
We refer to them as \emph{$\{0,\frac{1}{2}\}$-cuts}.
The \emph{$\{0,\frac{1}{2}\}$-closure} of $P$
is defined as
\[
  P_{\frac{1}{2}}(A,b) := 
  \left\{ x \in P \colon 
  u^T Ax \le \lfloor u^T b \rfloor,\, u \in \{0,\tfrac{1}{2}\}^m,\, u^T A \in \Z^n \right\}.
\]
Note that $P_{\frac{1}{2}}(A,b)$ depends on the system $Ax \le b$ defining the polyhedron $P$.
From the definition, it follows that $P_I \subset P' \subset P_\frac{1}{2}(A,b) \subset P$.

$\{0,\frac{1}{2}\}$-cuts are prominent in polyhedral combinatorics;
examples of classes of inequalities that can be derived as $\{0,\frac{1}{2}\}$-cuts include the blossom inequalities 
for the matching polytope \cite{Edm,Chv} and the odd-cycle inequalities for the stable set polytope \cite{GS}.
Both classes of inequalities can be separated in polynomial time \cite{GS,PR}.
In general, though, separation (and, thus, optimization) over the $\{0,\frac{1}{2}\}$-closure of polyhedra is NP-hard:
Caprara and Fischetti \cite{CF} show that the following \emph{membership problem} for the $\{0,\frac{1}{2}\}$-closure is strongly coNP-complete (see also \cite[Theorem 2]{Eis}).
\begin{quote}
  Given $A \in \Z^{m \times n}, b \in \Z^m$ and $\hat{x} \in \Q^n$ such that $\hat{x} \in P := \{ x \in \R^n \colon Ax \le b \}$,
  decide whether $\hat{x} \in P_\frac{1}{2}(A,b)$.
\end{quote}
The membership problem remains strongly coNP-complete even when $Ax \le b$ defines a polytope in the 0/1 cube, 
as shown by Letchford, Pokutta and Schulz \cite{LPS}.
It is, however, well known that testing membership in the Gomory-Chv\'atal closure belongs to NP $\cap$ coNP 
if restricted to polyhedra $P$ with $P' = P_I$ (see, e.g., \cite{BP}),
which naturally includes all polyhedra $P$ whose $\{0,\frac{1}{2}\}$-closure coincides with $P_I$.
For instance, the relaxation of the matching polytope given by nonnegativity and degree constraints has this property:
If we add the blossom inequalities, the resulting linear system is sufficient to describe the integer hull \cite{Edm},
and it is even totally dual integral (TDI) \cite{CM}.
This motivates the following research questions that are the subject of this paper:
What is the computational complexity of recognizing rational polyhedra whose $\{0,\frac{1}{2}\}$-closure coincides with the integer hull,
and of deciding whether adding all $\{0,\frac{1}{2}\}$-cuts produces a TDI system? 

Related questions for the Gomory-Chv\'atal closure have been studied by Cornu\'ejols and Li \cite{CL}.
They prove that, given a rational polyhedron $P$ with $P_I = \emptyset$, deciding whether $P' = \emptyset$ is weakly NP-complete.
This immediately implies weak NP-hardness of verifying $P'=P_I$.
Cornu\'ejols, Lee and Li \cite{CLL} extend these hardness results to the case when $P$ is contained in the 0/1 cube.
Moreover, they show that deciding whether a constant number of Gomory-Chv\'atal inequalities is sufficient to obtain the 
integer hull is weakly NP-hard, even for polytopes in the 0/1 cube.
In this paper, we establish analogous hardness results for the $\{0,\frac{1}{2}\}$-closure.
Our main result is the following theorem, where $\allone$ denotes the all-one vector.
\begin{theorem} \label{thm:main}
Given $A \in \Z^{m \times n}$ and $b \in \Z^m$ with $P := \{ x \in \R^n \colon Ax \le b \} \subset [0,1]^n$,
deciding whether $P_{\frac{1}{2}}(A,b) = P_I$ is strongly NP-hard, 
even when the inequalities $-x \le 0$ and $x \le \allone$ are part of the system $Ax \le b$.
\end{theorem}

We give a proof of this theorem in the next section.
Our proof implies several further hardness results, which we explain in \cref{sec:corollaries}. 
In particular, deciding whether adding all $\{0,\frac{1}{2}\}$-cuts to a given linear system $Ax \le b$ produces a TDI system,
is strongly NP-hard.
We also establish strong NP-hardness of the following problems:
deciding whether the $\{0,\frac{1}{2}\}$-closure coincides with the Gomory-Chv\'atal closure; 
deciding whether a constant number of $\{0,\frac{1}{2}\}$-cuts suffices to obtain the integer hull.
Finally, we give a hardness result for the membership problem for the $\{0,\frac{1}{2}\}$-closure, 
which is slightly stronger than the one of Letchford, Pokutta and Schulz \cite{LPS}. 


\section{Proof of Theorem 1} \label{sec:proof}

\begin{proof}[Proof of \cref{thm:main}]
We reduce from \textsc{Stable Set}:
\begin{quote}
  Let $G=(V,E)$ be a graph
  and $k \in \N, k \ge 2$. Does $G$ have a stable set of size at least $k$?
\end{quote}
It is well known that \textsc{Stable Set} is strongly NP-hard \cite{Kar}. 
Note that the problem remains strongly NP-hard if restricted to graphs with minimum degree at least 2:
Given an instance of \textsc{Stable Set} specified by $G$ and $k$, we construct a new graph $G'$ by adding two dummy nodes
to $G$ as well as all edges with at least one endpoint being a dummy node. Every node in $G'$ has degree at least 2, 
and every stable set in $G'$ of size $k \ge 2$ is a stable set in $G$ of the same size.

Consider an instance of \textsc{Stable Set} given by $G=(V,E)$ and $k \ge 2$.
By the above observation, we may assume that every node in $V$ has degree at least 2.
Note that $|V|=:n \ge 3$ and $|E|=:m \ge 3$ in this case.
Let $A := 2 \cdot \allone \allone^T - M^T$ where
$M \in \{0,1\}^{m \times n}$ denotes the edge-node incidence matrix of $G$ and $\allone$ is the all-one vector of appropriate dimension.
We define a polytope $P \subset \R^m$ by the following system of inequalities:
\begin{align}
  0 \;\le\; x &\;\le\; \allone      \label{eq:P_1} \\
  Ax &\;\le\; 2 \cdot \allone       \label{eq:P_A} \\
  (2k-3) \allone^T x &\;\ge\; 2k-3  \label{eq:P_bottom}
\end{align}

\begin{claim} \label{claim:integral}
$P_I = \{ x \in P \colon \allone^T x = 1 \}$.
\end{claim}
\begin{proof}[Proof of \cref{claim:integral}] \isclaim
If we add all inequalities in \cref{eq:P_A}, we obtain the valid inequality $2(n-1) \allone^T x \le 2n$.
Every integral point $x$ in $P$ therefore satisfies $\allone^T x = 1$. 
Since $A \in \{1,2\}^{n \times m}$, it is easy to check that every unit vector is indeed contained in $P$.
We conclude that
\[
  P_I = \{ x \in [0,1]^m \colon \allone^T x = 1 \} \supset \{ x \in P \colon \allone^T x = 1 \} \supset P_I.  \qedhere
\]
\end{proof}

The $\{0,\frac{1}{2}\}$-cuts that can be derived from \cref{eq:P_1,eq:P_A,eq:P_bottom} 
are all the inequalities of the following two types with $u \in \{0,\frac{1}{2}\}^n$ and $v \in \{0,\frac{1}{2}\}^m$:
\begin{align} 
  \sum_{i=1}^m \left( 2u^T \allone + \left\lfloor v_i - (Mu)_i \right\rfloor \right) x_i 
    &\le 2u^T \allone + \left\lfloor v^T \allone \right\rfloor
    \label{eq:cut} \\
  \sum_{i=1}^m \left( 2u^T \allone - (k-1) + \left\lfloor \tfrac{1}{2} + v_i - (Mu)_i \right\rfloor \right) x_i 
    &\le 2u^T \allone - (k-1) + \left\lfloor \tfrac{1}{2} + v^T \allone \right\rfloor
    \label{eq:cut_bottom}
\end{align}
The first type \cref{eq:cut} defines all cuts that are derived only from \cref{eq:P_1,eq:P_A},
whereas the second type \cref{eq:cut_bottom} also uses inequality \cref{eq:P_bottom}.
The vector $u$ is the vector of multipliers for inequalities \cref{eq:P_A} while $v$ collects the multipliers for the upper bounds in \cref{eq:P_1}.

In what follows, $P_\frac{1}{2}$ denotes the $\{0,\frac{1}{2}\}$-closure of $P$ defined by \cref{eq:P_A,eq:P_1,eq:P_bottom}
together with \cref{eq:cut,eq:cut_bottom} for all $u \in \{0,\frac{1}{2}\}^n$ and $v \in \{0,\frac{1}{2}\}^m$.

\begin{claim} \label{claim:iff}
$P_{\frac{1}{2}} = P_I$ if and only if there is a $\{0,\frac{1}{2}\}$-cut equivalent to $\allone^T x \le 1$.
\end{claim}
\begin{proof}[Proof of \cref{claim:iff}] \isclaim
If there is such a cut, then $P_{\frac{1}{2}} \subset \{ x \in P \colon \allone^T x \le 1 \} = P_I$ by \cref{claim:integral}.
To see the ``only if'' part, consider the vector $y = (\frac{1}{n} + \varepsilon) \allone$ 
for some small $\varepsilon>0$. Clearly, $y \notin P_I$ since $\allone^T y > 1$. 
We claim that there is a choice for $\varepsilon$ such that $y \in P$ and
$y$ satisfies all $\{0,\frac{1}{2}\}$-cuts except those that are equivalent to $\allone^T x \le 1$.
First observe that every cut (of either type \cref{eq:cut} or \cref{eq:cut_bottom}) as well as every inequality in \cref{eq:P_bottom,eq:P_A}
may be written as $a^T x \le \alpha$ for some $a \in \Z^m, \alpha \in \Z$ 
where $a_i \le \alpha$ for all $i \in [m]$ and $\alpha \le m+n$.
If $\alpha \le 0$, we clearly have $a^T y \le \alpha$ since $y \ge \tfrac{1}{m} \allone$.
If $\alpha > 0$ and $a^T x \le \alpha$ is not equivalent to $\allone^T x \le 1$, 
then $a_i < \alpha$ for at least one $i \in [m]$.
It follows that
$
  a^T y \le \alpha - \tfrac{1}{m} + \varepsilon (m \alpha -1).
$
For instance, taking $\varepsilon := \frac{1}{m^2 (m+n)}$ 
yields $a^T y \le \alpha$ as desired.
\end{proof}

In particular, the proof of \cref{claim:iff} shows that the inequality $\allone^T x \le 1$ is not valid for $P$.

\begin{claim} \label{claim:char}
No cut of type \cref{eq:cut} is equivalent to $\allone^T x \le 1$.
\end{claim}
\begin{proof}[Proof of \cref{claim:char}] \isclaim
Let $u \in \{0,\frac{1}{2}\}^n$ and $v \in \{0,\frac{1}{2}\}^m$.
If $u=0$, \cref{eq:cut} is dominated by the sum of the inequalities $\left\lfloor v_i - (Mu)_i \right\rfloor x_i \le 0$ for all $i \in [m]$. Note that these are valid for $P$ since $\left\lfloor v_i - (Mu)_i \right\rfloor \le 0$ for all $i \in [m]$.
If $v=0$, the cut \cref{eq:cut} is a trivial cut which is only derived from inequalities in the description of $P$ with even right-hand sides. Hence, we may assume that both $u \ne 0$ and $v \ne 0$.
It suffices to show that $\left\lfloor v_i - (Mu)_i \right\rfloor < \left\lfloor v^T \allone \right\rfloor$ for at least one $i \in [m]$. If $v^T \allone \ge 1$, there is nothing to show.
Now let $v^T \allone = \frac{1}{2}$ and suppose for the sake of contradiction that 
$\left\lfloor v_i - (Mu)_i \right\rfloor \ge 0$ for all $i \in [m]$. It follows that $Mu \le v$.
Since every column of $M$ has at least two nonzero entries by assumption, we obtain $u=0$, a contradiction.
\end{proof}

\begin{claim} \label{claim:char_bottom}
A cut of type \cref{eq:cut_bottom} induced by $u \in \{0,\frac{1}{2}\}^n$ and $v \in \{0,\frac{1}{2}\}^m$
is equivalent to $\allone^T x \le 1$ if and only if $v = 0$, $2Mu \le \allone$, and $2u^T \allone \ge k$.
\end{claim}
\begin{proof}[Proof of \cref{claim:char_bottom}] \isclaim
Suppose first that $v \ne 0$. Then, for every $i \in [m]$, we have
$
  \left\lfloor \tfrac{1}{2} + v_i - (Mu)_i \right\rfloor \le 1 \le \left\lfloor \tfrac{1}{2} + v^T \allone \right\rfloor.
$
This holds with equality for all $i \in [m]$ simultaneously only if $v_i = \frac{1}{2}$ and $v^T \allone \le 1$, 
contradicting $m \ge 3$. Thus, no inequality of the form \cref{eq:cut_bottom} with $v \ne 0$ has identical coefficients that coincide with the right-hand side.
We may therefore assume that $v=0$.
 
If $2u^T \allone \le k-1$, inequality \cref{eq:cut_bottom} is redundant: It is the sum of the inequalities
$(2u^T \allone - (k-1)) \allone^T x \le 2u^T \allone - (k-1)$ and 
$\left\lfloor \tfrac{1}{2} - (Mu)_i \right\rfloor x_i \le 0$ for all $i \in [m]$, all of which are valid for $P$.
Assuming that $2 u^T \allone \ge k$, inequality \cref{eq:cut_bottom} is equivalent to $\allone^T x \le 1$
if and only if $(Mu)_i \le \frac{1}{2}$ for all $i \in [m]$.
\end{proof}

Putting together \cref{claim:iff,claim:char,claim:char_bottom}, we conclude that
$P_{\frac{1}{2}} = P_I$ if and only if there exists 
some $u \in \{0,\frac{1}{2}\}^n$ such that $2u$ is the incidence vector of a stable set in $G$ of size at least $k$.
\end{proof}


\section{Further hardness results} \label{sec:corollaries}

A careful analysis of the proof of \cref{thm:main}
shows that, if the polytopes $P$ constructed in the reduction satisfy $P_\frac{1}{2} = P_I$, 
there is a single $\{0,\frac{1}{2}\}$-cut that certifies this (see \cref{claim:iff}).
This observation immediately implies the following corollary.

\begin{corollary}
Let $k \in \N$ be a fixed constant.
Given $A \in \Z^{m \times n}$ and $b \in \Z^m$ with $P := \{ x \in \R^n \colon Ax \le b \} \subset [0,1]^n$,
deciding whether one can obtain $P_I$ by adding at most $k$ $\{0,\frac{1}{2}\}$-cuts is strongly NP-hard, 
even when $k=1$, and $-x \le 0$ and $x \le \allone$ are part of the system $Ax \le b$.
\end{corollary}

Moreover, let us remark that $P'=P_I$ for the polytopes $P$ arising from the reduction. 
This follows from the fact that for $n \ge 3$,
the inequality $\allone^T x \le \left\lfloor 2n / 2(n-1) \right\rfloor = 1$ is a Gomory-Chv\'atal cut for $P$, see the proof of \cref{claim:integral}.

\begin{corollary}
Given $A \in \Z^{m \times n}$ and $b \in \Z^m$ with $P := \{ x \in \R^n \colon Ax \le b \} \subset [0,1]^n$,
deciding whether $P_{\frac{1}{2}}(A,b) = P'$ is strongly NP-hard, 
even when $-x \le 0$ and $x \le \allone$ are part of the system $Ax \le b$.
\end{corollary}

The linear systems arising from our reduction have another interesting property.
The inequality description \cref{eq:P_1,eq:P_A,eq:P_bottom,eq:cut,eq:cut_bottom} of $P_{\frac{1}{2}}$ in the proof
of \cref{thm:main} is a TDI system if and only if $P_{\frac{1}{2}} = P_I$. 
This can be seen as follows.
Since any polyhedron defined by a TDI system with integer right-hand sides is integral \cite{EG}, it suffices to show the ``if'' part.
Suppose that $P_{\frac{1}{2}} = P_I$. By the proof of \cref{thm:main}, there exist vectors $u',u'' \in \{0,\frac{1}{2}\}^n$
such that $2Mu' \le \allone$, $2Mu'' \le \allone$, $2(u')^T \allone = k$, and $2(u'')^T \allone = k-2 \ge 0$ (see \cref{claim:char_bottom}).
The cuts of type \cref{eq:cut_bottom} derived with $u'$ and $u''$ (where we take $v=0$) 
are the inequalities $\allone^T x \le 1$ and $-\allone^T x \le -1$, respectively. 
The system defined by these two inequalities and $x \ge 0$ 
is a subsystem of \cref{eq:P_1,eq:P_A,eq:P_bottom,eq:cut,eq:cut_bottom} that is sufficient to describe $P_\frac{1}{2}$ 
(see \cref{claim:integral}) and that is readily seen to be TDI:
Let $c \in \Z^m$. We can assume w.l.o.g.\ that $c_1$ is the largest coefficient of $c$. 
It follows that $\max \{ c^T x \colon x \in P_\frac{1}{2} \} = c_1$.
It suffices to show that the inequality $c^T x \le c_1$ is a nonnegative integer linear combination of the selected subsystem.
Indeed, it is the sum of $c_1 \allone^T x \le c_1$ 
(which is a nonnegative integer multiple of $\allone^T x \le 1$ or $-\allone^T x \le -1$)
and $-(c_1-c_i) x_i \le 0$ for all $i \in [m]$.
The above argument shows the following result.

\begin{corollary}
Let $A \in \Z^{m \times n}$ and $b \in \Z^m$. 
Deciding whether the system given by $Ax \le b$ and all $\{0,\frac{1}{2}\}$-cuts derived from it 
is TDI, is strongly NP-hard, 
even when $-x \le 0$ and $x \le \allone$ are part of the system $Ax \le b$.
\end{corollary}

Further note that the presence of the constraints $x \le \allone$ in \cref{eq:P_1} is not essential 
for our reduction in the proof of \cref{thm:main}. In fact, the upper bounds are redundant:
For every $i \in [m]$, consider a row of $A$ such that the entry in column $i$ is equal to 2. Such a row exists because $n \ge 3$. The corresponding inequality in \cref{eq:P_A} together with the nonnegativity constraints $-x_j \le 0$ (possibly twice) for all $j \ne i$ yields $2x_i \le 2$ for all $x \in P$.
As the only relevant cuts among \cref{eq:cut,eq:cut_bottom} are those with $v=0$, we conclude that all of the above results still hold true when the upper bounds $x \le \allone$ are not part of the input.

Another byproduct of our proof of \cref{thm:main} is that the membership problem for the $\{0,\frac{1}{2}\}$-closure
of polytopes in the 0/1 cube is strongly coNP-complete. This has already been shown by Letchford, Pokutta and Schulz \cite{LPS}.
However, neither of the two different reductions given in \cite{LPS} constructs linear systems that include both nonnegativity 
constraints and upper bounds on every variable. When these constraints are required to be part of the input, membership testing
remains strongly coNP-complete, as the following result shows.
\begin{corollary}
The membership problem for the $\{0,\frac{1}{2}\}$-closure of polytopes contained in the 0/1 cube
is strongly coNP-complete, even when the inequalities $-x \le 0$ and $x \le \allone$ are part of the input.
\end{corollary}
\begin{proof}
The problem clearly belongs to coNP.
To show hardness, we use the same reduction from \textsc{Stable Set} as in the proof of \cref{thm:main}.
The vector $y$ defined in the proof of \cref{claim:iff} satisfies $y \notin P_{\frac{1}{2}}$ 
if and only if the instance of \textsc{Stable Set} is a ``yes'' instance.
The encoding length of $y$ is polynomial in $m$ and $n$ if we choose $\varepsilon$ as in \cref{claim:iff}.
\end{proof}


\section{Concluding remarks}

It is worth pointing out that the problem of recognizing integrality of the $\{0,\frac{1}{2}\}$-closure is in coNP
when the membership problem for the $\{0,\frac{1}{2}\}$-closure can be solved in polynomial time:
If $P = \{ x \colon Ax \le b \}$ is a rational polyhedron with $P_{\frac{1}{2}}(A,b) \ne P_I$, 
it suffices to exhibit a fractional vertex $\hat{x}$ of $P_{\frac{1}{2}}(A,b)$ along with a corresponding basis. 
Then one can verify in polynomial time that $\hat{x} \in P_{\frac{1}{2}}(A,b)$ and that $\hat{x}$ is indeed a vertex.
This observation can be found in \cite[Chapter 9]{GLS} where it is stated
in the context of recognizing $t$-perfect graphs. These are the graphs
whose stable set polytope is determined by nonnegativity and edge constraints together with the 
odd-cycle inequalities \cite{Chv1}. In fact, the odd-cycle inequalities can be derived as $\{0,\frac{1}{2}\}$-cuts 
from the other two classes of inequalities \cite{GS}. This means that a graph is $t$-perfect if and only if the 
$\{0,\frac{1}{2}\}$-closure of the relaxation of its stable set polytope given by nonnegativity and edge constraints is integral.
Since a separating odd-cycle inequality can be found in polynomial time \cite{GS},
recognizing $t$-perfection is in coNP. Whether this problem is in NP or in P is not known (see \cite[Chapter 9]{GLS}).
However, some classes of $t$-perfect graphs are known to be polynomial-time recognizable, 
including claw-free $t$-perfect graphs \cite{BS} and bad-$K_4$-free graphs \cite{GSh}. Interestingly, for these two classes of graphs, the linear system in \cite{Chv1} that determines the stable set polytope is TDI \cite{BSt,Sch1}.
It is not known whether this holds true for $t$-perfect graphs in general (see \cite{Sch}).


\end{document}